\documentclass[a4paper, UKenglish]{lipics-v2018}

\usepackage{microtype}
\usepackage[T1]{fontenc}
\usepackage{textcomp}
\usepackage[inline,shortlabels]{enumitem} 

\usepackage{tikz}
\usetikzlibrary{decorations.pathmorphing,decorations.pathreplacing,arrows,
  automata,patterns,calc,3d}
\usepackage{amsthm,amsfonts,amssymb}

\makeatletter
\theoremstyle{plain}
\newtheorem*{rep@theorem}{\rep@title}
\newcommand{\newreptheorem}[2]{%
\newenvironment{rep#1}[1]{%
 \def\rep@title{#2 \ref{##1}}%
 \begin{rep@theorem}}%
 {\end{rep@theorem}}}
\makeatother

\newreptheorem{theorem}{Theorem}
\newreptheorem{corollary}{Corollary}

\usepackage{stmaryrd}
\usepackage{xcomment}
\usepackage{xspace}
\usepackage{complexity}

\renewcommand{\mid}{\;\middle\vert\; }
\newcommand{\ZZ}{\ensuremath{\mathbb{Z}}\xspace}

\newcommand{\NN}{\ensuremath{\mathbb{N}}\xspace}
\newcommand{\QQ}{\ensuremath{\mathbb{Q}}\xspace}

\newcommand{\alphabet}{{\Sigma}}
\newcommand{\fullshift}{\ensuremath{\alphabet^{\ZZ^d}}\xspace}

\newcommand{\norm}[1]{{\left\|#1\right\|}}

\newcommand{\pizu}{\ensuremath{\Pi^0_1}\xspace}
\newcommand{\sizu}{\ensuremath{\Sigma^0_1}\xspace}

\newcommand{\sizd}{\ensuremath{\Sigma^0_2}\xspace}

\newcommand{\sizn}{\ensuremath{\Sigma^0_n}\xspace}
\newcommand{\pizn}{\ensuremath{\Pi^0_n}\xspace}

\newcommand{\subshift}[1]{{{X}_{#1}}}
\newcommand{\forbid}{{\mathcal F}}
\newcommand{\periodset}[1]{{\ensuremath{P_{#1}}}}
\newcommand{\F}{\mathcal{F}}
\newcommand{\langu}[1]{{\ensuremath{\mathcal{L}\left(#1\right)}}}

\newcommand\ie{{\em i.e.}\xspace}
\DeclareMathOperator{\lcm}{lcm}

\title{Aperiodic points in $\ZZ^2$-subshifts}

\author{Anael Grandjean}{Laboratoire d'Algorithmique, Complexité et Logique\\
  Université Paris-Est Créteil, France}{anael.grandjean@u-pec.fr}{}{Sponsored by grant TARMAC ANR 12 BS02 007 01.}
\author{Benjamin Hellouin de Menibus}{Laboratoire de Recherche en Informatique,\\ Université Paris-Sud, CNRS, CentraleSupélec, Université Paris-Saclay, France}{hellouin@lri.fr}{ https://orcid.org/0000-0001-5194-929X}{}
\author{Pascal Vanier}{Laboratoire d'Algorithmique, Complexité et Logique\\
  Université Paris-Est Créteil, France}{pascal.vanier@lacl.fr}{}{Sponsored by grant TARMAC ANR 12 BS02 007 01.}
\authorrunning{A. Grandjean, B. Hellouin de Menibus and P. Vanier}
\Copyright{Anael Grandjean, Benjamin Hellouin de Menibus and Pascal Vanier}
\subjclass{\ccsdesc[500]{Theory of computation~Models of computation}}
\keywords{Subshifts of finite type, Wang tiles, periodicity,
  aperiodicity, computability, tilings}

\category{}

\relatedversion{hal-01722008v1}

\supplement{}

\funding{}

\nolinenumbers
\acknowledgements{The authors wish to thank anonymous reviewers for many helpful remarks and improvements.}

\EventEditors{Ioannis Chatzigiannakis, Christos Kaklamanis, D\'{a}niel Marx, and Don Sannella}
\EventNoEds{4}
\EventLongTitle{45th International Colloquium on Automata, Languages, and Programming (ICALP 2018)}
\EventShortTitle{ICALP 2018}
\EventAcronym{ICALP}
\EventYear{2018}
\EventDate{July 9--13, 2018}
\EventLocation{Prague, Czech Republic}
\EventLogo{eatcs}
\SeriesVolume{107}
\ArticleNo{496}

\begin{document}

\maketitle
\begin{abstract}
  We consider the structure of aperiodic points in $\mathbb Z^2$-subshifts, and in particular the positions at which they fail to be periodic. We prove that if a $\mathbb Z^2$-subshift contains points whose smallest period is arbitrarily large, then it contains an aperiodic point. This lets us characterise the computational difficulty of deciding if an $\mathbb Z^2$-subshift of finite type contains an aperiodic point. Another consequence is that $\mathbb Z^2$-subshifts with no aperiodic point have a very strong dynamical structure and are almost topologically conjugate to some $\mathbb Z$-subshift. Finally, we use this result to characterize sets of possible slopes of periodicity for $\mathbb Z^3$-subshifts of finite type.

\end{abstract}

A subshift on $\ZZ^d$ is a set of colorings of $\ZZ^d$ by a finite set of colors avoiding some family of forbidden patterns. When this family is finite, the subshift is called a \emph{subshift of finite type} (SFT). In dimension 2, SFTs are equivalent to sets of tilings by Wang tiles: Wang tiles are unit squares with colored borders that cannot be rotated and may be placed next to each other only if the borders match. 

Wang tiles were introduced by Wang in order to study the decidability of some fragments of logic \cite{WangI,WangII}. He thus introduced the Domino Problem: given a set of Wang tiles, do they tile the plane? (in other words, is the corresponding subshift nonempty?) Wang first conjectured that whenever a tileset tiles the plane, it can do so in a periodic manner, which would have implied the decidability of the Domino Problem.

In dimension $1$ the problem is decidable. A $\ZZ$-SFT corresponds to the set of biinfinite walks on some automaton and it tiles the line if and only if the automaton contains a cycle. Such a cycle provides a periodic point of the SFT, so non-empty \ZZ-SFTs always contain a periodic point.
The situation is dramatically different in higher dimension. Berger \cite{BergerPhD} proved that there exists tilesets in dimension 2 that tile the plane only aperiodically, and that the Domino Problem was therefore undecidable.

Thus, from the start, periodicity and aperiodicity have been at the heart of the study of Wang tiles and SFTs, and the main tool in understanding their structural properties and the answer to various decision problems. To give a few examples:
\begin{itemize}
   \item The presence of a dense set of periodic points is related to the decidability of the problem of deciding whether a given pattern appears in some point of an SFT \cite{ks}.
\item The finite subshifts on $\ZZ^d$ are exactly the subshifts containing only
  periodic configurations with $d$ non-colinear vectors of periodicity \cite[Theorem 3.8]{structural}. These configurations can be seen as finite configurations.
  This result has recently been extended to subshifts on groups \cite{SaloMey}. 
\item Countable SFTs always contain a finite configuration
  and if they are not finite, then they contain a configuration
  with exactly one vector of periodicity \cite[Theorem 3.11]{structural}.
\item A subshift always contains a quasiperiodic
  configuration \cite{Birkhoff,Du}, a configuration in which every finite pattern appears in any window of sufficiently large size depending only on the size of
  the pattern.
\end{itemize}

In this article we study the structure of aperiodic points in $\ZZ^2$-SFTs, and in particular the repartition of the coordinates where it avoids to be periodic. Our main result is:
\begin{theorem}\label{thm:main}
  There exists a computable function $f$ that satisfies the following. Assume $X$ is a $\ZZ^2$-subshift such that for any finite set of periods $\periodset{}$, $X$ contains a configuration that avoids all periods of $\periodset{}$. Then $X$ contains an aperiodic point that avoids every period $p$ at distance at most $f(\norm{p})$ from $0$.
\end{theorem}

This means that aperiodicity can be ``organised'' in concentric balls around a common center, in such a way that a proof of aperiodicity for any vector may be found near this center. As a consequence, when a subshift does not contain any aperiodic point, it must have a finite number of directions of periodicity:

\begin{corollary}\label{cor:main}
For any subshift $X$ with no aperiodic point, there is a finite set of periods $\periodset{}$ such that any configuration of $X$ is periodic for some period $p\in\periodset{}$.
\end{corollary}

This will lead to a further characterization of subshifts containing no aperiodic points in Section~\ref{S:charnoaper}.

These results have a variety of consequences. Gurevich and Koryakov~\cite{Gure} proved that for $d\geq 2$ it is undecidable to know whether an SFT contains a periodic, resp. aperiodic configuration. While it is easy to see that checking whether an SFT contains a periodic configuration is a recursively enumerable problem ($\sizu$ in the arithmetical hierarchy), it remained an open problem whether deciding if an SFT contains an aperiodic configuration was even in the arithmetical hierarchy. One of the consequences of Theorem~\ref{thm:main} is that it is $\pizu$.

Periodicity is also a central topic of symbolic dynamics since sets of periods and directions of periodicity constitute conjugacy invariants. For example, we prove that a $\ZZ^2$-subshift with no aperiodic point has a very strong dynamical structure and is essentially equivalent (almost topologically conjugate) to some $\ZZ$-subshift, and this is true for SFTs as well. In particular, various classical decision problems are decidable for this class, its topological entropy is 0 and its entropy dimension is at most 1.
Sets of periods have also been studied and characterized through computability and complexity theory \cite{persft}. \cite{MoutotVanier} recently proved that any \sizd set of $\left(\QQ\cup\{\infty\}\right)^2$ can be realized as a set of slopes of a $\ZZ^3$-SFT. Another consequence of Theorem~\ref{thm:main} is that this becomes a characterization.

The article is organized as follows: Section~\ref{S:defs} recalls some definitions and notations, Section~\ref{S:mainthm} is devoted to the proof of Theorem~\ref{thm:main}, Section~\ref{S:consequences} is devoted to its consequences and Section~\ref{S:counter} shows a counter example for $\ZZ^d$
subshifts when $d\geq 3$.

\section{Definitions}\label{S:defs}
Throughout the paper, we consider the distance on $\ZZ^d$ defined by the uniform norm $d(i,j) = ||i-j||_\infty$.

\subsection{Subshifts}
We provide here standard definitions about subshifts, which may be found in greater detail in \cite{LindMarcus}.

The $d$-dimensional full shift is the set $\Sigma^{\ZZ^d}$ where $\Sigma$ is a finite alphabet whose elements are called \emph{letters} or \emph{symbols}. Each element of $\Sigma^{\ZZ^d}$ is called a \emph{configuration}  or \emph{point}. A configuration may be seen as a coloring of $\ZZ^d$ with the letters of $\Sigma$.
For $v\in\ZZ^d$, the \emph{shift function} $\sigma_v:\Sigma^{\ZZ^d}\to\Sigma^{\ZZ^d}$ is defined by $\sigma_v(x)_z = x_{z+v}$. The \emph{full shift} equipped with the distance $d(x,y)=2^{-\min\left\{ \norm{v}\mid v\in\ZZ^d, x_v\neq y_v \right\}}$ forms a compact metric space on which the shift functions act as homeomorphisms.
A closed shift invariant subset $X$ of $\Sigma^{\ZZ^d}$ is called a \emph{subshift} or \emph{shift}.

A \emph{pattern} of shape $\Gamma$, where $\Gamma$ is a finite subset of $\ZZ^d$, is an element of $\Sigma^\Gamma$ or alternatively a function $p:\Gamma\to \Sigma$. A configuration $x$ \emph{avoids a pattern} $\gamma$ of shape $\Gamma$ if $\forall z\in\ZZ^d$, $\sigma_v(x)_{|\Gamma}\neq \gamma$ and \emph{contains} $\gamma$ if it does not avoid it.

For a family of forbidden patterns $\F$, denote $X_{\F}$ the set of configurations that avoid $\F$. Then $X_\F$ is a subshift, and every subshift can be defined in this way. When a subshift can be defined this way by a finite family, it is called a \emph{subshift of finite type}. When a subshift can be defined by a recursively enumerable family of forbidden patterns, it is called an \emph{effectively closed subshift}.

If $X$ is a subshift, we denote by $\langu{X}$ its \emph{language}, \ie the set of patterns that appear somewhere in one of its points. 

\begin{definition}[Periodicity]
    A configuration $x$ is \emph{periodic} of period $v$ if there exists $v\in\ZZ^d\setminus\{(0,0)\}$ such that $\forall z\in\ZZ^d, x_z = x_{z+v}$. More precisely, a configuration is \emph{$k$-periodic} if it has exactly $k$ linearly independent periods. If a configuration has no period, then it is said to be \emph{aperiodic}. A subshift is
    \emph{aperiodic} if all its points are aperiodic.
    
\end{definition}

Denote by $B(z,n)$ the ball of radius $n$ centered in $z \in\ZZ^d$.

Let $x\in\fullshift$ and $p \in \ZZ^2$. If there exists $z\in\ZZ^2$ such that $x_{z} \neq x_{z+p}$, we say that $x$ \emph{avoids} period $p$. The pair $(z,z+p)$ is called an \emph{avoidance} of period $p$ in configuration $x$. We say that a configuration \emph{avoids a set of periods} $\periodset{}$ if it avoids every period
in $\periodset{}$.

Let $\periodset{}$ be a set of periods. We denote $\periodset{}'$ the set obtained from $\periodset{}$ by replacing each period $p$ by the least commun multiple of all periods of $\periodset{}$
that are colinear to $p$. More formally : $\periodset{}' = \left\{ \lcm ( q \mid \text{$q \in \periodset{}$ and $q$ and $p$ are colinear}) \mid p \in \periodset{} \right\} $.
Observe that $\periodset{}'$ is a set of pairwise non-colinear periods.

Except in the last section, the subshifts we will be considering will
implicitely be $\ZZ^2$-subshifts.
\subsection{Arithmetical hierarchy}
We give some basic definitions used in computability theory and in particular about the arithmetical hierarchy. More details may be found in \cite{Rogers}.

Usually the arithmetical hierarchy is seen as a classification of sets according to their logical characterization.
For our purpose we use an equivalent definition in terms of computability classes and Turing machines with oracles:

  \begin{itemize}
  \item $\Delta^0_0 = \Sigma_0^0 = \Pi_0^0$ is the class of recursive (or computable) problems.
    \item \sizn is the class of recursively enumerable (RE) problems with an oracle $\Pi_{n-1}^0$.
    \item \pizn the complementary of \sizn, or the class of co-recursively enumerable (coRE) problems with an oracle $\Sigma_{n-1}^0$.
    \item $\Delta^0_n = \sizn \cap \pizn$ is the class of recursive (R) problems with an oracle $\Pi_{n-1}^0$.
  \end{itemize}

In particular, \sizu is the class of recursively enumerable problems and \pizu is the class of
co-recursively enumerable problems.

\section{Main result}\label{S:mainthm}
This whole section is dedicated to the proof of the Theorem~\ref{thm:main} and Corollary~\ref{cor:main}.
Given a subshift that contains an aperiodic point, we prove that it contains some aperiodic point where all period avoidances are organised in concentric balls around a common center, in such a way that each period $p$ is in a ball whose radius only depends on $\norm{p}$. This result is used in a compactness argument to prove that, if a subshift contains configurations whose smallest period is arbitrarily large, then it contains an aperiodic point.

Actually, our algorithm can only gather avoidances in a small ball if all the periods are non-colinear. Fortunately we can easily build a set $\periodset{}'$.

\begin{lemma}
Let $\periodset{}$ be a set of periods. Any configuration avoiding $\periodset{}'$ also avoids $\periodset{}$.
\end{lemma}  
\begin{proof}
  Each period $p$ in $\periodset{}$ has an integer multiple $p'\in\periodset{}'$. Each
  avoidance of $p'$ induces an avoidance of $p$.
\end{proof}
\begin{lemma}
Let $\periodset{}$ be a set of pairwise non-colinear periods. Let $x$ be a
configuration avoiding $\periodset{}$. Then $x$ avoids 
every period of $\periodset{}$ in some ball of radius $\sum_{p \in \periodset{}}\norm{p}$.
\label{lem:ball}
\end{lemma}

\begin{proof}
  We prove the result by induction on the number of periods in $\periodset{}$. When $\periodset{}$ is
  a singleton the case is trivial. Now suppose $\periodset{}$ is not a
  singleton. Denote $p_0, p_1,\dots, p_n$ the periods in $\periodset{}$. By induction hypothesis, we can find a ball $B_{n-1}$ of radius $\sum_{i<n}
    \norm{p_i}$ centered in $b_{n-1}$ containing avoidances for every period
  in $\periodset{}$ except $p_n$.
  Similarly, we find a ball $B'_{n-1}$ of radius $\sum_{i > 0} \norm{p_i}$, centered
  in $b'_{n-1}$ that contains avoidances of every period in 
  $\periodset{}$  except $p_0$.
  We now show that either an avoidance of $p_0$ exists near a copy of $B'_{n-1}$
  or an avoidance of $p_n$ exists near a copy of $B_{n-1}$.

  Consider the ball $B(b_{n-1}+p_n,\sum_{i < n}\norm{p_i})$, the translated image of $B_{n-1}$ by the vector $p_n$. Either $x_{z}=x_{z + p_n}$ for all $z\in B_{n-1}$, \ie{}
  the two balls are filled the same way, or we found an avoidance $(z,z+p_n)$
  with $z\in B_{n-1}$. In the latter case, the result is proved.

  As depicted in Figure~\ref{fig:gather}, this process can be iterated for both $B'_{n-1}$ and $B_{n-1}$ until either we find the
  necessary avoidance or the centers of the balls are close to each other: Since
  $p_0$ and $p_n$ are not colinear, and assuming $\norm{p_n} \geq \norm{p_0}$, there exists $i,j\in\ZZ$ such that $\norm{b_{n-1}
    + ip_n - b'_{n-1} + jp_0} \leq \norm{\frac{p_n}{2}} + \norm{\frac{p_0}{2}} < \norm{p_n}$.
	We thus found a ball centered in $b'_{n-1}+j p_0$ and of radius $\sum_{i}\norm{p_i}$ containing the two balls we translated, and therefore an  avoidance of each period in $\periodset{}$. Denote $B_n$ this new ball and $b_n$ its center.
\end{proof}
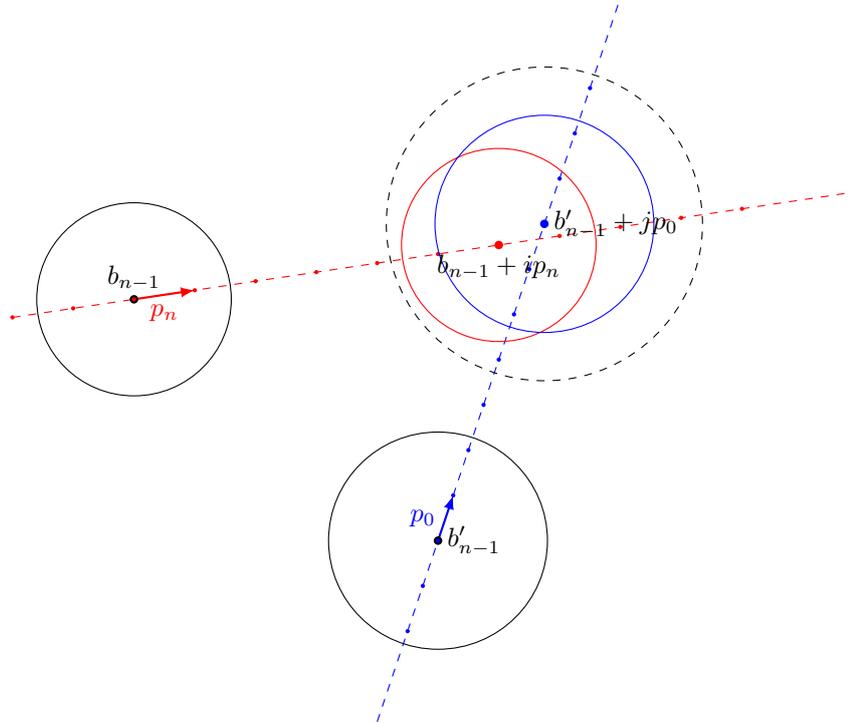
\begin{figure}[htb]
    \centering
	\begin{tikzpicture}[scale=0.8]
	\draw[dashed, red] (0,4.7) -- (14,6.8);
	\draw[dashed, blue] (6,-2) -- (10,10);
	\draw[thick,red,->,>=latex] (2,5) -- node[anchor = north]{$p_n$} (3,5.15);
	\draw[thick,blue,->,>=latex] (7,1) -- node[anchor = east]{$p_0$} (7.25,1.75);
	\draw[anchor = south] (2,5) node{$b_{n-1}$};
	\fill (2,5) circle(2 pt); 
	\draw (2,5) circle(1.6 cm);
	\draw[anchor = west] (7,1) node{$b'_{n-1}$};
	\fill (7,1) circle(2pt);
	\draw (7,1) circle(1.8 cm);
	\foreach \x in {-2,-1,...,10}
		\fill[red] (2+1*\x,5+0.15*\x) circle(1 pt);
	\foreach \x in {-2,-1,...,10}
		\fill[blue] (7+0.25*\x,1+0.75*\x) circle(1 pt);
	\draw[red] (8,5.9) circle(1.6cm); 
	\draw[blue] (8.75,6.25) circle(1.8cm);
	\fill[red] (8,5.9) circle(2pt);
	\fill[blue] (8.75,6.25) circle(2pt);
	\draw[anchor = west] (8.75,6.25) node{$b'_{n-1} + jp_0$};
	\draw[anchor = north] (8,5.9) node{$b_{n-1} + ip_n$};
	\draw[dashed] (8.75,6.25) circle(2.6cm);

	\end{tikzpicture}
	\caption{
    The process of translating $B_{n-1}$ and $B'_{n-1}$ close to each other: each
    translation may uncover the desired avoidance and if not, the two balls next to each other necessarily do so.
	}
	\label{fig:gather}
\end{figure}
In the previous proof, the distance between $b_{n-1}$ and $b_n$ only depends on $p_0$, $p_n$ and the distance between $b_{n-1}$ and $b'_{n-1}$.

Therefore there is a computable function $f(p_0, p_n ,r)$ such that, if $b_{n-1}$ and $b'_{n-1}$ belong to a common ball $B(z,r)$, then $\norm{b_n-z}\leq f(p_0, p_n ,r)$.

\begin{lemma}\label{lem:ballprime}

Let $\periodset{} = \{p_0, \dots, p_n\}$ be a set of non-colinear periods. Define $f'(\periodset{}, r)$ recursively as:\begin{itemize}
\item if $n=0, f'(\periodset{}, r) = r$;
  \item if $n>1, f'(\periodset{}, r) = f\big(p_0, p_n, \max\big[f'(\periodset{}\backslash\{p_0\}, r), f'(\periodset{}\backslash\{p_n\}, r)\big]\big)$
  \end{itemize}

 Take $x\in X$, and assume that $x$ avoids every period $p\in\periodset{}$ in some ball $B(z,r)$. Then $x$ avoids every period $p\in\periodset{}$ in some ball $B(z', \sum_\periodset{}\norm{p})$, with $\norm{z'-z}\leq f'(\periodset,\,r)$.
\end{lemma}

\begin{proof}
  We prove the lemma by induction. If $n=0$, the result is obvious.
  

  Now assume $n>1$. By applying the induction hypothesis twice on $p_0,\dots,p_{n-1}$ and $p_1, \dots, p_n$, we find two balls $B_{n-1} = B(b_{n-1}, \sum_{i<n}\norm{p_i})$
  and $B'_{n-1} = B(b'_{n-1}, \sum_{i>0}\norm{p_i})$ such that $\norm{b_n-z}\leq f'(\periodset{}\backslash\{p_n\}, r)$ and $\norm{b'_n-z}\leq f'(\periodset{}\backslash\{p_0\}, r)$. Applying Lemma~\ref{lem:ball} on these balls, we obtain the desired ball with $\norm{b_n - z}\leq f'(\periodset{}, r)$.
\end{proof}

The next lemma states that in a configuration avoiding periods $p_1, \dots, p_n$, we can organise the avoidances in concentric balls around a common center, so that the distance of the avoidance of any given period from the center does not depend on $n$ but only on the period itself.
	
 \begin{lemma}\label{lem:concentric}
   Let $\periodset{n} = \{p_0, \dots, p_n\}$ be a set of periods. Denote $P_i = \{p_0, \dots, p_i\}$ for $i\leq n$. Define recursively a function $g$ such that
\[ g(\{p\}) =\norm{p}\quad\text{and}\quad g(\periodset{n}) = g(\periodset{n-1}) + f'\left(\periodset{n-1}', \sum_{\periodset{n}'}\norm{p}\right) + \sum_{\periodset{n}'}\norm{p}\]
Take $x$ a point that avoids $\periodset{n}'$ in a ball $B(z,\sum_{\periodset{n}'}\norm{p})$. There exist $z'\in\ZZ^2$ such that: 
   \begin{itemize}
   \item $\norm{z' - z} \leq g(\periodset{n})$
   \item x avoids $\periodset{i}$ in the ball $B(z',g(\periodset{i}))$ for any $i \leq n$.
   \end{itemize}
   
    \end{lemma}
\begin{proof}
  We proceed by induction on $n$. 

If $n=0$, then since $B(z,\norm{p_0})$ contains an avoidance of $\periodset{0} = \{p_0\}$, taking $z'$ to be this avoidance satisfies the requisite.

Assume $n>0$. Since $B(z,\sum_{\periodset{n}'}\norm{p})$ contains avoidances of every period in $\periodset{n}'$, it contains avoidances of every period in $\periodset{n-1}'$. Indeed, if some period $q$ is in $\periodset{n-1}'$ but not in $\periodset{n}'$, then by construction a multiple of $q$, say $Mq$, is in $\periodset{n}'$. Now if $(z_q, z_q+Mq)$ is an avoidance of $Mq$, at least one of $(z_q+mq, z_q+(m+1)q)$ for $m<M$ is an avoidance of $q$ and is contained in the same ball.

Applying Lemma~\ref{lem:ballprime} on $\periodset{n-1}'$, we find a ball $B(z_0, \sum_{\periodset{n-1}'}\norm{p})$ that contains avoidances for all periods in $\periodset{n-1}'$ and such that $\norm{z_0-z}\leq f'(\periodset{n-1}', \sum_{\periodset{n}'}\norm{p})$. 

Now apply the induction hypothesis on this ball, obtaining $z'$ such that $\norm{z' - z_0} \leq g(\periodset{n-1})$ and for any $i\leq n-1$, the ball $B(z',g(\periodset{i}))$ contains avoidances of every period in $\periodset{i}$. This inductive process is depicted in Figure~\ref{fig:heredite}.

By the triangular inequality, $\norm{z'-z} \leq g(\periodset{n-1}) + f'(\periodset{n}', \sum_{\periodset{n-1}'}\norm{p})$. Since $g(P_n) \geq \norm{z-z'}+\sum_{\periodset{n}'}\norm{p}$, $B(z', g(\periodset{n}))$ contains entirely the ball $B(z,\sum_{\periodset{n}'}\norm{p})$. Therefore $B(z', g(\periodset{n}))$ avoids $P_n$ and $\norm{z'-z} \leq g(\periodset{n})$, proving the lemma.

  \begin{figure}
    \centering
    \begin{tikzpicture}[scale=0.9]
      \begin{scope}
      \clip (-2.5,2.5) rectangle (10.5,-6.7);
      \draw[red] (0,0) circle(2cm); 
      \draw[dashed,red] (0,-1.4) -- node[anchor=west]{$p_6$} (0.5,1);
      \draw[dashed,blue] (-0.5,-0.3) --node[anchor=south]{$p_5$} (-2,0);
      \draw[dashed,blue] (0.5,1) --node[anchor=south]{$p_4$} (1.75,1);
      \draw[dashed,blue] (0,-1.4) --node[anchor=east]{$p_3$} (-0.5,-0.3);
      \draw[dashed,blue] (-0.5,-0.3) --node[anchor=west]{$p_2$} (-0.2,0.6);
      \draw[dashed,blue] (0,-1.4) --node[anchor=east]{$p_1$} (-0.3,-2);
      \draw[dashed,blue] (0,-1.4) --node[anchor=north]{$p_0$} (0.4,-1.2);
	
	  \draw[thick,>=latex,->] (1.95,0.4) -- (3.57,0.72);
	
      \draw[blue] (5,1) circle(1.5cm); 
      \draw[dashed,blue] (5.5,0.7) --node[anchor=north west]{$p_0$} (5.9,0.9); 
      \draw[dashed,blue] (5.9,0.9) --node[anchor=west]{$p_3$} (5.4,2);
      \draw[dashed,blue] (5.5,0.7) --node[anchor=north west]{$p_1$} (5.2,0.1); 
      \draw[dashed,blue] (4,1) --node[anchor=west]{$p_2$} (4.3,1.9);
      \draw[dashed,red] (5.5,0.7) --node[anchor=south]{$p_5$} (4,1); 
      \draw[dashed,blue] (3.95,0.1) --node[anchor=south]{$p_4$} (5.2,0.1);
	
	  \draw[thick,>=latex,->] (5.85,-0.25) -- (6.4,-1);
	
      \draw[blue] (7,-1.9) circle(1.1cm); 
      \draw[dashed,blue] (7.25,-1.9) --node[anchor=north]{$p_0$} (7.65,-1.7); 
      \draw[dashed,blue] (6.4,-1.7) --node[anchor=west]{$p_3$} (6.9,-2.8);
      \draw[dashed,blue] (6.4 ,-1.7) --node[anchor=east]{$p_1$} (6.7,-1.1); 
      \draw[dashed,blue] (7.25,-1.9) --node[anchor=west]{$p_2$} (7.55,-1); 
      \draw[dashed,red] (6.4,-1.7) --node[anchor=south]{$p_4$} (7.65,-1.7);
	
	  \draw[thick,>=latex,->] (6.32,-2.74) -- (5.95,-3.1);
	  
      \draw[blue] (5.4,-3.8) circle(0.9cm); 
      \draw[dashed,blue] (5.2,-4.5) --node[anchor=north]{$p_0$} (5.6,-4.3); 
      \draw[dashed,red] (5.6,-4.3) --node[anchor=south west]{$p_3$} (5.1,-3.2); 
      \draw[dashed,blue] (4.8 ,-3.8) --node[anchor=east]{$p_1$} (5.1,-3.2); 
      \draw[dashed,blue] (5.6,-4.3) --node[anchor=west]{$p_2$} (5.9,-3.4);
      
      \draw[thick,>=latex,->] (6.3,-3.95) -- (6.68,-4.05);
      
      \draw[blue] (7.4,-4.2) circle(0.75cm); 
      \draw[dashed,blue] (7.1,-3.9) --node[anchor=south]{$p_0$} (7.5,-3.7); 
      \draw[dashed,blue] (7.4,-4.8) --node[anchor=west]{$p_1$} (7.7,-4.2); 
      \draw[dashed,red] (7.1,-3.9) --node[anchor=east]{$p_2$} (7.4,-4.8);
	
	  \draw[thick,>=latex,->] (7.75,-4.95) -- (7.78,-5);

      \fill[black] (8,-5.5) circle(0.05cm);
      \draw[blue] (8,-5.5) circle (0.6cm); 
      \draw[dashed,red] (8,-5.4) --node[anchor=north west]{$p_0$} (8.4,-5.2); 
      \draw[dashed,red] (8,-5.4) --node[anchor=north west]{$p_1$} (7.7,-6);
      
      \draw[dashed] (8,-5.5) circle (0.7cm);
      \draw[dashed] (8,-5.5) circle (2.1cm);
      \draw[dashed] (8,-5.5) circle (3.7cm);
      \draw[dashed] (8,-5.5) circle (5.7cm);
      \draw[dashed] (8,-5.5) circle (8cm);
      \draw[dashed] (8,-5.5) circle (10.5cm);
      
      \draw[->,>=latex] (8,-5.5) -- +(170:0.7cm);
      \draw (8,-5.5) +(135:0.4) node{$g(\periodset{1})$};
      \draw[->,>=latex] (8,-5.5) -- +(173:2.1cm);
      \draw (8,-5.5) +(163:1.4) node{$g(\periodset{2})$};
      \draw[->,>=latex] (8,-5.5) -- +(176:3.7);
      \draw (8,-5.5) +(171:2.9) node{$g(\periodset{3})$};
      \draw[->,>=latex] (8,-5.5) -- +(179:5.7);
      \draw (8,-5.5) +(176:4.7) node{$g(\periodset{4})$};
      \draw[->,>=latex] (8,-5.5) -- +(182:8);
      \draw (8,-5.5) +(179:6.85) node{$g(\periodset{5})$};
      \draw[->,>=latex] (8,-5.5) -- +(185:10.5);
      \draw (8,-5.5) +(182:9.25) node{$g(\periodset{6})$};
      
      \end{scope}
    \end{tikzpicture}
    
    \caption{The red ball contains avoidance for $\periodset{n}'$. Each dashed line represents an avoidance. We then consider the set of blue lines, $\periodset{n-1}'$. We gather these avoidances in a close smaller ball with Lemma~\ref{lem:ballprime}. We repeat this process until there are only two avoidances in the ball. The red avoidances are disposed in the way we wanted, around the center of the smallest blue ball.}
    \label{fig:heredite}
  \end{figure}
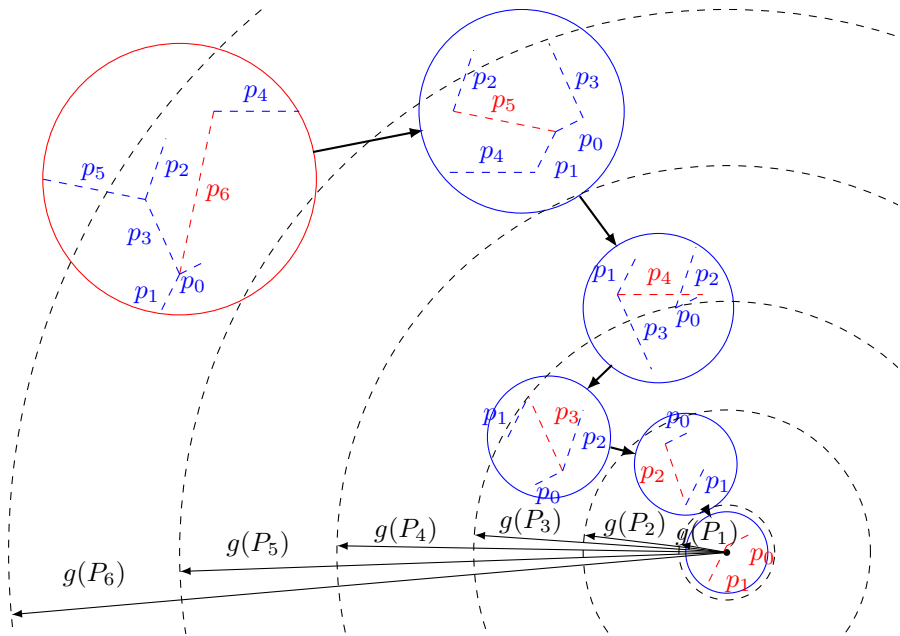

\end{proof}

\begin{reptheorem}{thm:main}
  Let $X$ be a $\ZZ^2$-subshift. 
  Assume that for every finite set of periods $\periodset{}$, $X$ contains a configuration that avoids all periods of $\periodset{}$.
  
  Then $X$ contains an aperiodic configuration that has an avoidance of every period $p$ at distance at most $g'(\norm{p}) = g(B(0,\norm{p})')$ from $0$, where $g$ is the function defined in Lemma~\ref{lem:concentric}.
\end{reptheorem}

\begin{remark}
$g'$ is not polynomial because $B(0,n)'$ contains exponentially large vectors. Since our bounds are likely very far from sharp, we leave the exact computation to the reader.
\end{remark}

\begin{proof}
Take $B(0,n)$ the set of periods of norm $n$ or less and $x_1, \dots, x_n, \dots$ a sequence of configurations such that $x_n$ avoids $B(0,n)'$.

By applying Lemma~\ref{lem:ball} on $B(0,n)'$, we obtain for each $x_n$ a ball of radius $\sum_{B(0,n)'}\norm{p}$ that avoids $B(0,n)'$. Applying Lemma~\ref{lem:concentric} on $B(0,n)$, we get that $x_n$ avoids all periods in $B(0,i)'$ in some ball $B(z_n,g(B(0,i))'))$ for all $i\leq n$. Since $X$ is compact, any limit point of the sequence $(\sigma_{z_n}(x_n))_{n\in\NN}$ is in $X$, and it avoids all periods in $B(0,i)$ in the ball $B(0,g(B(0,i)'))$ for all $i\in\NN$. It is in particular aperiodic.\end{proof}

\begin{repcorollary}{cor:main}
For any subshift $X$ with no aperiodic point, there is a finite set of periods $\periodset{}$ such that any configuration of $X$ is periodic for some period $p\in\periodset{}$.
\end{repcorollary}

\begin{proof} This is a simple reciprocal to Theorem~\ref{thm:main}.
\end{proof}

\section{Consequences}\label{S:consequences}

\subsection{Existence of an aperiodic configuration is \pizu}
\begin{corollary}
  The following problem is \pizu-computable:
  \begin{description}
  \item[Input] A finite set of forbidden patterns $\F$.
  \item[Output] Does the $\ZZ^2$-SFT $\subshift\forbid$ contain an aperiodic configuration?
  \end{description}
  \label{thm:pi1aper}
\end{corollary}
\begin{proof}
  Let  $(p_i)_{i\in\NN}$ be an enumeration of all possible periods and $P_n=\{p_0,\dots,p_n\}$.
  Theorem~\ref{thm:main} gives us a bound on the size of the patterns in which
  to look for avoidances of each period. For each $n\in\NN$ in order, the algorithm enumerates all patterns on a ball of radius $g(P_n)$ that do not contain a forbidden pattern, and check if one of them contains avoidances for every period of $P_k$ ($k\leq n$) in the ball of radius $g(P_k)$ in its center. If such a pattern does not exist for some $n$, it means that either the SFT is empty or that all its points are periodic for some period $p$ with $\norm{p}\leq n$.

  Assume the algorithm runs infinitely. For every $k$, there exists some $n\ge k$ such that if a pattern on the ball of radius $g(P_n)$ does not contain a forbidden pattern, then the subpattern on the ball of radius $g(P_k)$ is in the language of $X$. Therefore, for each $P_k$ we find a pattern in $\L(X)$ that avoids all periods of $P_k$, and we conclude by Theorem~\ref{thm:main}.
\end{proof}
\subsection{Structure of subshifts without aperiodic points}\label{S:charnoaper}

In this subsection, we consider notions from dynamical system theory. A dynamical system is given by a pair $(C, \Phi)$ where $C$ is a compact set and $\Phi:C\to C$ is a continuous function.

\begin{definition}[Topological conjugacy]
Let $(C_1, \Phi_1)$ and $(C_2, \Phi_2)$ be dynamical systems.

  $(C_1, \Phi_1)$ and $(C_2, \Phi_2)$ are \emph{topologically conjugate} if there exists a continuous bijection $\pi : C_1\to C_2$ such that $\pi\circ \Phi_1 = \Phi_2\circ \pi$.

  $(C_1, \Phi_1)$ and $(C_2, \Phi_2)$ are \emph{almost topologically conjugate} if there exists $(C_3, \Phi_3)$ and continuous surjections $\pi_i: C_3\to C_i$ that are bijective almost everywhere\footnote{Except for a finite set of points.} such that $\pi_i\circ \Phi_3 = \Phi_i\circ \pi_i$ for $i=1,2$.
\end{definition}

See \cite{LindMarcus} or \cite{Kitchens} for more information on topological conjugacy and almost conjugacy in the context of symbolic dynamics. We need slightly more general definitions since we consider subshifts of different dimensions.

Notice that we can have $(C_2, \Phi_2) = (C_3, \Phi_3)$ (and $\pi_2 = id$) in the last definition; this is the case in the next proof.

\begin{theorem}\label{thm:conjugacy}
Let $X$ be a two-dimensional subshift with no aperiodic point. There exists a vector $v$ and a one-dimensional subshift $Y$ such that $(X, \sigma_v)$ is almost topologically conjugate to $(Y, \sigma)$.
  
If $X$ is of finite type, then $Y$ can be chosen of finite type as well.
\end{theorem}

\begin{proof}
  Let $X$ be a two-dimensional subshift of finite type with no aperiodic point. By Corollary~\ref{cor:main}, there is a finite set of periods $\periodset{}$ such that any configuration of $X$ is periodic of some period $p \in \periodset{}$. We assume that $\periodset{}$ does not contain two colinear periods, by taking their least common integer multiple if necessary.

  For the clarity of the argument, we assume in the following that $\periodset{}$ does not contain any period colinear to $(0,1)$. Since $\periodset{}$ is finite, the proof can be adapted for a different vector.\bigskip

  Take $p = (p^0, p^1) \in \periodset{}$, assuming $p^0 > 0$, and denote $X_p = \{x\in X : x\text{ admits }p\text{ as a period.}\}$. $X_p$ is a closed set and it is a classical argument (see for instance \cite[\textsection 2.1.2]{persft} or \cite[Lemma 5.2]{structuring}) that $(X_p, \sigma_{(0,1)})$ is topologically conjugate to a one-dimensional SFT, which we repeat here for completeness. Define:
  \[\pi_p = \left\{\begin{array}{ccc}
  \Sigma^{\ZZ^2} &\to& (\Sigma^{p^0-1})^\ZZ\\
  x&\mapsto& ((x_{i,j})_{0\leq i< p^0})_{j\in\ZZ}
  \end{array}\right.\]
  Denote $Y_p = \pi_p(X_p)$. It is not hard to see that $\pi_p$ is a one-to-one continuous function between $X_p$ and $Y_p$ and that $Y_p$ is a subshift of finite type if $X_p$ is,  since it can be defined by a finite recoding of the forbidden patterns of $X_p$. Furthermore, $\pi_p\circ \sigma_{(0,1)} = \sigma\circ \pi_p$, so it is a topological conjugacy between $(X_p, \sigma_{(0,1)})$ and $(Y_p, \sigma)$.

  For any $p_1\neq p_2 \in\periodset{}$, $X_{p_1}\cap X_{p_2}$ is a set of $2$-periodic configurations that admit non-colinear periods $p_1$ and
  $p_2$; there are a finite number of such configurations, so $|X_{p_1}\cap
  X_{p_2}|<+\infty$. In other words, $X = \bigcup_{p\in\periodset{}} X_p$ and
  the union is disjoint except for a finite set of configurations.

  Denote $Y = \sqcup_{p\in\periodset{}} Y_p$ (disjoint union). $Y$ is a subshift
  on the alphabet $\sqcup_{p\in\periodset{}} \Sigma_p$, where $\Sigma_p$ is the
  alphabet of $Y_p$. Furthermore, $Y$ is of finite type if every $Y_p$ is of
  finite type.

  Define $\varphi : Y \to X$ by $\varphi|_{Y_p} = \pi_p^{-1}$. We can check that $\varphi$ is surjective and almost everywhere bijective, and that $\varphi\circ \sigma_1 = \sigma_{(0,1)}\circ \varphi$. We have proved that $(Y, \sigma_1)$ is almost topologically conjugate to $(X, \sigma_{(0,1)})$.
\end{proof}

\subsection{Various properties of subshifts with no aperiodic points}
Theorem~\ref{thm:conjugacy} implies that the property of having no aperiodic point gives a very strong structure to a subshift. This is particularly the case for subshifts of finite type, where many problems that are indecidable in dimension 2 are completely solved in dimension 1, and these solutions carry through almost topological conjugacy.

In this section, we make use of notations that were defined in the proof of Theorem~\ref{thm:conjugacy}: $X_p, Y_p, \pi_p$ and $\varphi$.\bigskip

\subsubsection{Decision problems} Decision problems have been a staple of the theory of multidimensional subshifts of finite type: the seminal paper of Wang proved that the emptiness problem (given a list of forbidden patterns $\F$, is $X_\F = \emptyset$?) is decidable for two-dimensional non-aperiodic subshifts of finite type, but Berger later proved that the problem was undecidable without this assumption \cite{BergerPhD}. We consider other classical decision problems: the extension problem, which is undecidable for multidimensional subshifts of finite type (as a consequence of the above), and the injectivity and surjectivity problems, which are undecidable even on the two-dimensional full shift \cite{KariRevCA}. 

A $\ZZ^d$-\emph{cellular automaton} is a continuous function $F : \Sigma^{\ZZ^d} \to \Sigma^{\ZZ^d}$ that commutes with every shift function. It can be defined equivalently by a local rule $f : \Sigma^\Gamma\to\Sigma$ for a finite shape $\Gamma$ by $F(x)_v = f(\sigma_v(x)_{\Gamma})$ for all $v\in\ZZ^d$.

\begin{corollary}\label{cor:decision}
  The following problems are decidable for two-dimensional subshifts of finite type with no aperiodic point:
  \begin{description}
  \item[Extension problem] given a list of forbidden patterns $\F$ and a pattern $w$, do we have $w\in\langu{X_\F}$?
  \item[Injectivity / surjectivity problem] given a list of forbidden patterns $\F$ and a cellular automaton $\Phi : \Sigma^{\ZZ^2} \to \Sigma^{\ZZ^2}$, is $\Phi|_{X_\F}$ injective? surjective on $X_\F$?
  \end{description}
\end{corollary}

Links between periodic points and the above problems have already been considered in \cite{ks, Fiorenzi}. 

\begin{proof}[Proof sketches.]
  \begin{description}
    \item [Extension problem] Assume $w$ has shape $[-n,n]^2$. By Theorem~\ref{thm:conjugacy} we have $X = \varphi(Y)$ where $\varphi$ is continuous on a compact space, hence uniformly continuous. In other words, for every $n$, there exists $r$ such that the value of $\varphi(y)_{[-n,n]^2}$ only depends on $y_{[-r,r]}$. Since the extension problem is decidable on one-dimensional subshifts of finite type, enumerate all words $v\in \mathcal L(Y)$ and check whether $w=\varphi(v)$ for some $v$.

    \item[Injectivity problem] By Corollary~\ref{cor:main}, there is a finite set of periods $\periodset{}$
such that $X_{\mathcal F} = \bigcup_{p\in\periodset{}}X_p$.
 
    A cellular automaton is injective if and only if it is reversible. Since the image of a configuration of period $p$ by a cellular automaton also has period $p$, we have $\Phi|_{X_\F}(X_p) \subset X_p$ and (in the injective case) $\Phi|_{X_\F}^{-1}(X_p) \subset X_p$. It follows that $\Phi|_{X_\F}$ is injective if, and only if, $\Phi|_{X_p}$ is injective for every period $p\in\periodset{}$.
    
    Let $\pi_p : X_p\to Y_p$ be the continuous bijection defined in the proof of Theorem~\ref{thm:conjugacy}. $\pi_p\circ \Phi|_{X_p} \circ \pi_p^{-1}$ is a CA on $Y_p$ and it shares injectivity with $\Phi|_{X_\F}$. Injectivity of CA is decidable for one-dimensional subshifts of finite type \cite{Fiorenzi}.

\item[Surjectivity problem] Surjectivity is more delicate as a point in $X_p\cap X_q$ can be the image of a point from $X_p$ or $X_q$. However, $\Phi|_{X_\F}$ is surjective if and only if:
    \begin{enumerate}
    \item $\forall p\in\periodset{},\ \forall x \in X_{p} \backslash \bigcup_{p'\neq p} X_{p'},\ x\in\Phi|_{X_\F}(X_{p})$;
    \item $\forall p \neq p'\in\periodset{},\ \forall x \in X_{p}\cap X_{p'},\ \exists p''\in\periodset{},\ x\in\Phi|_{X_\F}(X_{p''})$, 
    \end{enumerate}
   Denote $X^\cap$ the finite set of configurations in case 2 (notice that we do not necessarily have $p'' =p$ or $p'' = p'$ if $x$ admits other periods as well). 

As in the previous case, we translate these properties on $Y_p$ and $\Phi_p = \pi_p\circ \Phi|_{X_p} \circ \pi_p^{-1}$ to work on $\ZZ$-SFT. 
Following \cite{Fiorenzi}, we can describe $\Phi_p(Y_p)$ by a finite automaton. 

For case 1, add to the finite automaton describing $\Phi_p(Y_p)$ an independent cycle for each element of $X^\cap$ and determine whether the resulting automaton describe the same SFT as the automaton describing $Y_p$. This algorithm is explained in \cite{LindMarcus}, Section~3.4.    

For case 2, since $x\in X^\cap$ is $2$-periodic, $\pi_p(x)$ is periodic, and it is easy to check by hand whether some $\Phi_p(Y_p)$ accepts $x$. Do this for all $x\in X^\cap$.
      \end{description}
  \end{proof}

  \begin{remark}
    If we did not know that $X$ admits a finite set of periods, the first proof would still show that the extension problem is in $\Sigma_1^0$ (RE). Since it is easy to show that it is in $\pizu$ (co-RE), our main result is technically unnecessary here.
  \end{remark}
    
\subsubsection{Topological entropy} 
Topological entropy is a widely-used parameter in information theory
(channel capacity) and dynamical systems theory (conjugacy invariant). Entropy
dimension is a more refined notion for systems of entropy zero, introduced in
\cite{Carvalho} and mainly used for multidimensional subshifts
\cite{Meyerovitch}.

\begin{corollary}\label{cor:entropy}
  Any two-dimensional subshift $X$ with no aperiodic point has zero topological
  entropy. Its entropy dimension is at most one.
\end{corollary}

\begin{proof}[Proof sketch.]
By Corollary~\ref{cor:main}, there is a finite set of periods $\periodset{}$
such that $\langu{X} = \bigcup_{p\in\periodset{}}\langu{X_p}$.
Consider a pattern $w$ of shape $[0,n-1]^2$ in $\langu{X_p}$, assuming for clarity that $p=(p^0, p^1)$ with $p^0\geq 0$ and $p^1\geq 0$. Since $w$ cannot contain an avoidance for $p$, it is entirely determined by its $p^0$ bottommost rows and $p^1$ leftmost columns. Therefore there are at most $(p^0+p^1)n$ such patterns. A similar argument applies when $p^0 < 0$ or $p^1 < 0$.

It follows that there are at most $\sum_p(|p^0|+|p^1|)n$ patterns of shape
$[0,n-1]^2$ in $\langu{X}$, proving the statement.
\end{proof}

\subsubsection{Density of periodic points}
Density of periodic points is a typical question in dynamical systems, for example when studying chaos in the sense of Devaney. See \cite{Fiorenzi} for more details, including a proof that two-dimensional subshifts of finite type do not have dense $2$-periodic points in general, even under an additional irreducibility hypothesis.

$X$ is \emph{irreducible} (or \emph{transitive}) if for any two patterns $\gamma_1,\ \gamma_2 \in \langu{X}$ of shapes $\Gamma_1$ and $\Gamma_2$ respectively, there exists $x\in X$ and two coordinates $v_1, v_2$ such that $\sigma_{v_1}(x)_{\Gamma_1} = \gamma_1$ and $\sigma_{v_2}(x)_{\Gamma_2} = \gamma_2$.

\begin{corollary}\label{cor:periodic}
Any irreducible two-dimensional subshift of finite type $X$ with no aperiodic point has dense $2$-periodic points.
\end{corollary}

\begin{proof}[Proof sketch.]
By Corollary~\ref{cor:main}, consider $\periodset{}$ a finite and minimal set of periods such that any configuration of $X$ is periodic for some period $p\in\periodset{}$. If $\periodset{}$ is not a singleton, take $p_1\neq p_2\in\periodset{}$. There exists two finite patterns $\gamma_1$ and $\gamma_2$ that contain an avoidance of $p_1$ and $p_2$, respectively (otherwise, $\periodset{}$ would not be minimal). By irreducibility, there exists $x\in X$ where $\gamma_1$ and $\gamma_2$ both appear, and therefore $x$ avoids every $p\in\periodset{}$, a contradiction. Therefore $\periodset{}$ is a singleton $\{p\}$ and $X = X_p$ is conjugated to $Y_p$. One-dimensional irreducible subshifts of finite type, such as $Y_p$, have dense periodic points (\cite{Fiorenzi}, Proposition 9.1). The image by $\pi_p^{-1}$ of a periodic point in $Y_p$ is a $2$-periodic point in $X$, from which the statement follows.
\end{proof}

\subsection{The full caracterization of slopes of $\ZZ^3$-SFTs}

Intuitively, slopes of a subshift are the directions that some configuration admits as a unique direction of periodicity. More formally:

\begin{definition}
  Let $X$ be a $\ZZ^d$-subshift. $\theta\in\left( \QQ\cup{\infty}
  \right)^{d-1}$ is a \emph{slope of periodicity} of $X$ if there exists a
  configuration $x\in X$ and a vector $v\in\ZZ^d$ such that:
  \begin{itemize}
      \item $v\ZZ = \left\{ v' \mid \sigma_{v'}(x)=x \right\}$
      \item and $\theta_i=v_1/v_{i+1}$, for all $i\in\{0,\dots,d-1\}$.
  \end{itemize}
\end{definition}

The set of slopes of periodicity of a subshift is a conjugacy invariant.  A consequence of Corollary~\ref{thm:pi1aper} 
is that the sets of slopes of periodicity of $\ZZ^3$-SFTs is
a \sizd-computable set, and together with \cite{MoutotVanier} this implies the following
caracterization:

\begin{theorem}
  \sizd-computable subsets of $S\subseteq\left( \QQ\cup\{\infty\} \right)^2$
  are exactly the sets realizable as sets of slopes of 
  $\ZZ^3$-subshifts of finite type. 
\end{theorem}
\begin{proof}
We know from \cite{MoutotVanier} that one can realize any such \sizd set $S$ as
a  set of slopes of a $\ZZ^3$-subshift. Let us now show the remaining direction. 

Given a slope $\theta$ and a set of forbidden patterns $\F$ as an input, we want to check whether
there exists a configuration in $X_\F$ whose vectors of periodicity all have direction $\theta$.

Using the notations of the proof of Theorem~\ref{thm:conjugacy}, for any $p\in\ZZ^2$, the set $X_p$ of configurations of period $p$ (for some $k>0$) can be seen as a $\ZZ^2$-SFT $Y$ computable from $\F$ and $p$.

There is a smallest vector $p_\theta$ such that all vectors of direction $\theta$ are integer multiples of it. Remark that $\theta$ is a slope of periodicity of $X_\F$ is and only if $X_{kp_\theta}$ contains an aperiodic configuration for some $k>0$. By Corollary~\ref{thm:pi1aper}, checking whether $X_{kp_\theta}$ contains an aperiodic configuration for a given $k$ is \pizu-computable. 
Therefore checking whether there is a $k>0$ for which this holds is \sizd-computable.
\end{proof}
	
\section{Counterexample for dimensions $d>2$}\label{S:counter}

We build a counterexample to Theorem~\ref{thm:main} in higher dimension. Take $\Sigma = \{0,1\}$ and define $X$ as follows:
\begin{itemize}
\item All symbols $1$ must form lines of direction vector $(1,0,0)$ (horizontal)
  or $(0,0,1)$ (vertical);
\item There is at most one vertical line;
\item All horizontal lines are repeated periodically with period $(0,0,n)$,
  where $n$ is the distance of the vertical line to any horizontal line.
\end{itemize}

In particular, if there is no vertical line, then there is at most one horizontal line. To sum up, a subshift configuration can be : \begin{enumerate*}\item all zeroes, \item one horizontal line, \item one vertical line, or \item the situation depicted in Figure~\ref{fig:stretcher}.\end{enumerate*}
\begin{figure}[htbp]
  \centering
  \begin{tikzpicture}[x={(-30:1cm)},y={(210:1cm)},z={(90:1cm)},scale=.5]
    \foreach \j in {1,3,...,8}{
      \draw[canvas is yz plane at x=0,very thick] (-6,\j) -- ++(12,0);
      \draw[canvas is yx plane at z=\j,latex-latex,dashed,gray] (0,-5) -- ++(0,5) node[midway,above,gray] {$n$};
      \draw[canvas is yz plane at x=0,latex-latex,dashed,gray] (0,\j) -- ++(0,2) node[midway,right,gray] {$n$};
    }
    \draw[canvas is yz plane at x=-5,very thick] (0,-4.5) -- (0,-4.15);
    \draw[canvas is yz plane at x=-5,very thick] (0,-3.85) -- (0,-2.15);
    \draw[canvas is yz plane at x=-5,very thick] (0,-1.85) -- (0,-.15);
    \draw[canvas is yz plane at x=-5,very thick] (0,0.15) -- (0,1.85);
    \draw[canvas is yz plane at x=-5,very thick] (0,2.15) -- (0,8.5);

    \draw[canvas is xz plane at y=0, gray, ->] (-10,0) -- ++(0,1) node[at end, above, gray] {$z$};
    \draw[canvas is xz plane at y=0, gray, ->] (-10,0) -- ++(1,0) node[near end,above, gray] {$x$};
    \draw[canvas is yz plane at x=-10, gray, ->] (0,0) -- ++(1,0) node[near end,above, gray] {$y$};
  \end{tikzpicture}
  \caption{A typical configuration of $X$: a line of ones along $z$ at distance
    $n$ of an $(xy)$ plane of lines along $x$. The only other types of configurations of $X$ are the configurations containing either a single vertical line, a single horizontal line, or no line at all.}
  \label{fig:stretcher}
\end{figure}
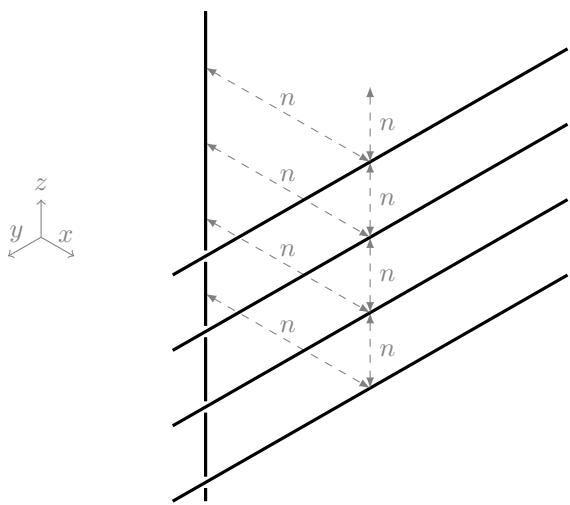

The configuration described in Figure~\ref{fig:stretcher} admits $(0,0,n)$ as period, and no shorter period. In particular, for every finite set of periods $\periodset{}$, $X$ contains a configuration that avoids $\periodset{}$ (by taking $n$ large enough). However, $\Sigma$ admits no aperiodic point\footnote{Notice that there cannot be a configuration with a single horizontal line and a single vertical line, which would be aperiodic.}. This example can easily be generalised to any $d>3$ by considering a $\ZZ^d$-subshift $X'$ that contains a copy of $X$ in at most one coordinate, and $0$ everywhere else: that is, \[x\in X' \Leftrightarrow \forall j\in \ZZ^{d-3}, (x_{i,j})_{i\in\ZZ^3} \in X\text{ and }(\forall j_1\neq j_2, (x_{i,j_1})_{i\in\ZZ^3} = 0\text{ or }(x_{i,j_2})_{i\in\ZZ^3} = 0).\]

This proves that Theorem~\ref{thm:main} does not hold in any dimension $d>2$.

\section{Open problems}

We have made clear that our main result does not hold for subshifts of dimension
$d\geq 3$. We do not know, however, whether Theorem~\ref{thm:conjugacy} or
Corollary~\ref{cor:entropy} holds in higher dimension, since the counterexample
introduced in Section~\ref{S:counter} does not contradict these results.

This counterexample is a subshift containing points with arbitrarily large periods but no aperiodic point. We do not know whether such a counterexample with infinitely many directions of periodicity exist. Moreover, the structure of $d$-dimensional subshifts of finite type for $d\geq 3$ remains open; the existence of this counterexample suggests that a making use of the finite type hypothesis is
necessary in higher dimension.
\bibliographystyle{plainurl}
\bibliography{biblio}
\end{document}